\pdfoutput=1
\documentclass[11pt]{article}
\usepackage{xspace}
\DeclareMathAlphabet{\mathbbold}{U}{bbold}{m}{n}
\usepackage{amssymb,amsmath,amsthm,graphicx,bbm}
\usepackage{mathtools}
\usepackage[usenames,dvipsnames,svgnames,table]{xcolor}
\usepackage{thmtools, thm-restate}
\usepackage[margin=1.in]{geometry}
\usepackage{multirow,array}
\usepackage{colortbl}
\usepackage{microtype}
\usepackage[pagebackref]{hyperref}
\hypersetup{
    pdftitle={}, 
    pdfauthor={}, 
    colorlinks=true, 
    linkcolor=blue, 
    citecolor=blue, 
    filecolor=blue, 
    urlcolor=blue 
}

\usepackage{caption}

\renewcommand{\backref}[1]{}

\renewcommand{\backrefalt}[4]{%
\ifcase #1 %
\or
[p.\ #2]%
\else
[pp.\ #2]%
\fi}

\usepackage[capitalise,nameinlink]{cleveref}

\usepackage{algorithm}
\usepackage[noend]{algpseudocode}
\usepackage{tikz,tikz-qtree}
\usepackage{wrapfig}
\usepackage{enumitem}

\theoremstyle{plain}
\newtheorem{theorem}{Theorem}

\newtheorem{proposition}[theorem]{Proposition}
\newtheorem{lemma}[theorem]{Lemma}

\newtheorem{conjecture}{Conjecture}

\newtheorem{definition}[theorem]{Definition}

\theoremstyle{remark}

\theoremstyle{plain}

\def\Reals{{\mathbb{R}}} 

\renewcommand{\Pr}{\mathop{\bf Pr\/}}
\newcommand{\E}{\mathop{\bf E\/}}

\newcommand{\OR}{\textsc{or}}
\newcommand{\AND}{\textsc{and}}
\newcommand{\Parity}{\textsc{parity}}
\newcommand{\eps}{\varepsilon}

\renewcommand{\hat}{\widehat}
\renewcommand{\tilde}{\widetilde}

\newcommand{\B}{\{0,1\}}

\let\OldLambda\lambda
\let\lambda\relax
\DeclareMathOperator{\lambda}{\OldLambda}
\DeclareMathOperator{\D}{\mathsf{D}}
\DeclareMathOperator{\C}{\mathsf{C}}
\DeclareMathOperator{\R}{\mathsf{R}}
\DeclareMathOperator{\Q}{\mathsf{Q}}

\DeclareMathOperator{\RC}{\mathsf{RC}}

\DeclareMathOperator{\s}{\mathsf{s}}
\DeclareMathOperator{\bs}{\mathsf{bs}}
\DeclareMathOperator{\adeg}{\mathsf{\widetilde{deg}}}
\let\deg\relax
\DeclareMathOperator{\deg}{\mathsf{deg}}
\DeclareMathOperator{\Adv}{\mathsf{Adv}}
\DeclareMathOperator{\SA}{\mathsf{SA}}
\DeclareMathOperator{\MM}{\mathsf{MM}}
\DeclareMathOperator{\SWA}{\mathsf{SWA}}
\DeclareMathOperator{\GSA}{\mathsf{GSA}}
\DeclareMathOperator{\K}{\mathsf{K}}
\DeclareMathOperator{\Dom}{Dom}
\DeclareMathOperator{\tr}{tr}

\newcommand{\norm}[1]{\left\lVert#1\right\rVert}

\definecolor{conj}{HTML}{C2C0BF}
\definecolor{open}{HTML}{A31F34} 

\newcommand{\ct}[2]{%
\cellcolor{white}%
\begin{tabular}[t]{@{}c@{}}#1\\[-5pt]#2\end{tabular}%
}

\newcommand{\co}[2]{%
\cellcolor{conj!70}%
\begin{tabular}[t]{@{}c@{}}#1\\[-5pt]#2\end{tabular}%
}

\newcommand{\cc}[2]{%
\cellcolor{conj!70}%
\begin{tabular}[t]{@{}c@{}}#1\\[-5pt]#2\end{tabular}%
}

\newcommand{\smcite}[1]{{\scriptsize \cite{#1}}}
\newcommand{\newbound}[1]{\colorbox{open!50}{#1}}
\newcommand{\Huangbound}[1]{\colorbox{green!50}{#1}}

\title{Quantum Implications of Huang's Sensitivity Theorem}
\author{%
Scott Aaronson\footnote{Department of Computer Science, University of Texas at Austin. \texttt{aaronson@cs.utexas.edu}} 
\and 
Shalev Ben-David\footnote{University of Waterloo.
\texttt{shalev.b@uwaterloo.ca}} 
\and 
Robin Kothari\footnote{Microsoft Quantum and Microsoft Research. \texttt{robin.kothari@microsoft.com}}
\and 
Avishay Tal\footnote{%
Department of Electrical Engineering and Computer Sciences, University of California at Berkeley. \texttt{atal@berkeley.edu}}
}

\begin{document}
\maketitle
\begin{abstract}
Based on the recent breakthrough of Huang (2019), we show that for any total Boolean function $f$, the deterministic query complexity, $\D(f)$, is at most quartic in the quantum query complexity, $\Q(f)$: $\D(f) = O(\Q(f)^4)$. This matches the known separation (up to log factors) due to Ambainis, Balodis, Belovs, Lee, Santha, and Smotrovs (2017).
We also use the result to resolve the quantum analogue of the Aanderaa--Karp--Rosenberg conjecture. We show that if $f$ is a nontrivial monotone graph property of an $n$-vertex graph specified by its adjacency matrix, then $\Q(f)=\Omega(n)$, which is also optimal.
\end{abstract}

\section{Introduction}
Last year, Huang resolved a major open problem in the analysis of Boolean functions called the \emph{sensitivity conjecture}~\cite{Huang2019}, which was open for nearly 30 years \cite{NS94}. Surprisingly, Huang's elegant proof takes less than 2 pages---truly a ``proof from the book.'' 
Specifically, Huang showed that for any total Boolean function, which is a function $f:\B^n\to\B$, we have
\begin{equation}\label{eq:Huang}
    \deg(f) \leq \s(f)^2,
\end{equation}
where $\deg(f)$ is the real degree of $f$ and $\s(f)$ is the (maximum) sensitivity of $f$. These measures and other measures appearing in this introduction are defined in \Cref{sec:prelim}.

In this note, we describe some implications of Huang's resolution of the sensitivity conjecture to quantum query complexity.  We observe that Huang actually proves a stronger claim, in which $\s(f)$ in \cref{eq:Huang} can be replaced by $\lambda(f)$, a spectral relaxation of sensitivity that we define later. This observation has several implications for quantum query complexity.

We use this observation to settle the optimal relation between the deterministic query complexity, $\D(f)$, and quantum query complexity, $\Q(f)$, for total functions.
We know from the seminal results of Nisan~\cite{Nisan91}, Nisan and Szegedy~\cite{NS94} and Beals et al.\ \cite{BBCMW01} that any total Boolean function $f$ satisfies%
\footnote{%
This means that for total functions, quantum query algorithms can only outperform classical query algorithms by a polynomial factor. On the other hand, for partial functions, which are defined on a subset of $\B^n$, exponential and even larger speedups are possible.}
\begin{equation}\label{eq:Beals}
    \D(f)=O(\Q(f)^6).
\end{equation}
Grover's algorithm \cite{Grover96} shows that for the $\OR$ function, a quadratic separation between $\D$ and $\Q$ is possible. This was the best known quantum speedup for total functions until the work of Ambainis et al.\ \cite{ABB+15}, who constructed a total function $f$ with 
\begin{equation}\label{eq:ABB}
\D(f) = \tilde{\Omega}(\Q(f)^{4}).    
\end{equation}

In this note, we show that the quartic separation (up to log factors) in \cref{eq:ABB} is actually the best possible:
\begin{theorem}\label{thm:D vs. Q}
For all Boolean functions $f:\B^n \to \B$, we have 
$\D(f) = O(\Q(f)^4)$. 	
\end{theorem}

We deduce \Cref{thm:D vs. Q} as a corollary of a new tight quadratic relationship between $\deg(f)$ and $\Q(f)$:

\begin{theorem}\label{thm:deg vs. Q}
For all Boolean functions $f:\B^n \to \B$, we have 
$\deg(f) = O(\Q(f)^2)$. 	
\end{theorem}

Observe that \Cref{thm:deg vs. Q} is tight for the $\OR$ function on $n$ variables, whose degree is $n$ and whose quantum query complexity is $\Theta(\sqrt{n})$ \cite{Grover96,BBBV97}. Prior to this work, the best relation between $\deg(f)$ and $\Q(f)$ was a sixth power relation, $\deg(f) = O(\Q(f)^6)$, which follows from \cref{eq:Beals}.

As discussed earlier, our proof relies on the restatement of Huang's result (\Cref{thm:Huang}), showing that $\deg(f)\leq \lambda(f)^2$, where $\lambda(f)$ is the spectral relaxation of sensitivity defined in \Cref{sec:mainproof}. We then show that the measure $\lambda(f)$ lower bounds the original quantum adversary method of Ambainis~\cite{Amb02}, which in turn lower bounds $\Q(f)$.

We now show how \Cref{thm:D vs. Q} straightforwardly follows from \Cref{thm:deg vs. Q} using two previously known connections between complexity measures of Boolean functions.

\begin{proof}[Proof of \Cref{thm:D vs. Q} assuming \Cref{thm:deg vs. Q}]
In \cite{Mid04}, Midrijanis showed that for all total functions $f$, we have
\begin{equation}\label{eq:Mid}
\D(f) \leq \bs(f) \deg(f),
\end{equation} 
where $\bs(f)$ is the block sensitivity of $f$.
	
\Cref{thm:deg vs. Q} shows that $\deg(f) = O(\Q(f)^2)$. Combining the relationship between block sensitivity and approximate degree from \cite{NS94} with the results of \cite{BBCMW01}, we get that $\bs(f) =  O(\Q(f)^2)$. (This can also be proved directly using the lower bound method in \cite{BBBV97}.)

Combining these three inequalities yields $\D(f) = O(\Q(f)^4)$ for all total Boolean functions $f$.
\end{proof}

It remains to show the main result, \Cref{thm:deg vs. Q}, which we do in \Cref{sec:mainproof} using the proof of the sensitivity conjecture by Huang~\cite{Huang2019} and the spectral adversary method in quantum query complexity \cite{BSS03}.

In \Cref{sec:AKR}, we also use \Cref{thm:deg vs. Q} to prove the quantum analogue of the famous \textit{Aanderaa--Karp--Rosenberg conjecture}. Briefly, this conjecture is about the minimum possible query complexity of a nontrivial monotone graph property, for graphs specified by their adjacency matrices.

There are variants of the conjecture for different models of computation. For example, the randomized variant of the Aanderaa--Karp--Rosenberg conjecture, attributed to Karp~\cite[Conjecture 1.2]{SW86} and Yao~\cite[Remark (2)]{Yao77}, states that for all nontrivial monotone graph properties $f$, we have $\R(f) = \Omega(n^{2})$. Following a long line of work, the current best lower bound is $\R(f) = \Omega(n^{4/3} \log^{1/3}n)$ due to Chakrabarti and Khot~\cite{CK01}.

The quantum version of the conjecture was raised by Buhrman, Cleve, de Wolf, and Zalka~\cite{BCdWZ99}, who observed that the best one could hope for is $\Q(f) = \Omega(n)$, because the nontrivial monotone graph property ``contains at least one edge'' can be decided with $O(n)$ queries using Grover's algorithm~\cite{Grover96}.
Buhrman et al.~\cite{BCdWZ99} also showed that all nontrivial monotone graph properties $f$ satisfy $\Q(f)=\Omega(\sqrt{n})$. The current best lower bound is $\Q(f) = \Omega(n^{2/3}\log^{1/6}n)$, which was credited to Yao in \cite{MSS07}.
We resolve this conjecture by showing an optimal $\Omega(n)$ lower bound.

\begin{theorem}\label{thm:qAKR}
Let $f:\B^{\binom{n}{2}} \to \B$ be a nontrivial monotone graph property. Then $\Q(f)=\Omega(n)$.
\end{theorem}

\Cref{thm:qAKR} follows by combining \Cref{thm:deg vs. Q} with a known quadratic lower bound on the degree of monotone graph properties.

\subsection{Known relations and separations}

\setlength{\intextsep}{0pt}%
\setlength{\columnsep}{15pt}%
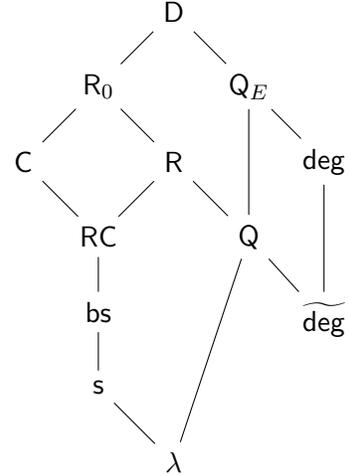
\begin{wrapfigure}{r}{0.35\textwidth}
\centering
\vspace{-2.5ex}
  \begin{tikzpicture}[x=1cm,y=1cm]

     \node (D) at(2,5){$\D$};
     \node (Rz) at(1,4){$\R_0$};
     \node (QE) at(3,4){$\Q_E$};
     \node (C) at(0,3){$\C$};
     \node (R) at(2,3){$\R$};
     \node (RC) at(1,2){$\RC$};
     \node (bs) at(1,1){$\bs$};
     \node (s) at(1,0){$\s$};
     \node (l) at(2,-1){$\lambda$};
     \node (deg) at(4,3){$\deg$};
     \node (Q) at(3,2){$\Q$};
     \node (adeg) at(4,0.9){$\adeg$};

     \path[-] (Rz) edge (D);
     \path[-] (QE) edge  (D);
     \path[-] (C) edge (Rz);
     \path[-] (R) edge (Rz);
     \path[-] (RC) edge (C);
     \path[-] (RC) edge (R);
     \path[-] (bs) edge (RC);
     \path[-] (bs) edge (s);
     \path[-] (s) edge (l);
     \path[-] (Q) edge (l);
     \path[-] (deg) edge (QE);
     \path[-] (Q) edge (QE);
     \path[-] (Q) edge (R);
     \path[-] (adeg) edge (Q);
     \path[-] (adeg) edge (deg);
  \end{tikzpicture}
    \caption{Relations between complexity measures. An upward line from a measure $M_1(f)$ to $M_2(f)$ denotes $M_1(f) = O(M_2(f))$ for all total functions $f$.\label{fig:rel}}
\end{wrapfigure}

\Cref{tab:sep} summarizes the known relations and separations between complexity measures studied in this paper (and more). This is an update to a similar table that appears in \cite{ABK16} with the addition of $\s(f)$ and $\lambda(f)$. 
Definitions and additional details about interpreting the table can be found in \cite{ABK16}.

For all the separations claimed in the table, we provide either an example of a separating function or a citation to a result that constructs such a function. All the relationships in the table follow by combining the relationships depicted in \Cref{fig:rel} and the following inequalities that hold for all total Boolean functions:

\begin{itemize}[noitemsep,topsep=4pt]
\item $\C(f) \leq \bs(f)\s(f)$ \cite{Nisan91}
\item $\D(f) \leq \bs(f)\C(f)$ \cite{BBCMW01}
\item $\D(f)\leq \bs(f)\deg(f)$ \cite{Mid04}
\item $\RC(f) = O(\adeg(f)^2)$ \cite{KT16}
\item $\R_0(f) = O(\R(f) \s(f) \log \RC(f))$ \cite{KT16}
\item $\deg(f) \leq \lambda(f)^2$ \cite{Huang2019}
\item $\s(f) \leq \lambda(f)^2$ (\Cref{lem:lambdalower})
\end{itemize}

\setlength{\tabcolsep}{2pt}
\renewcommand{\arraystretch}{1.3}
\begin{table}
\hspace{-1em}
\begin{minipage}{\linewidth}
\begin{center}
\caption{Best known separations between complexity measures}
\begin{tabular}{r||c|c|c|c|c|c|c|c|c|c|c|c}
{} & $\D$ & $\R_0$ & $\R$ & $\C$ & $\RC$ & $\bs$ & $\s$ & $\lambda$ & $\Q_E$ & $\deg$ & $\Q$ & $\adeg$ \\ \hline\hline

$\D$ &
\cellcolor{darkgray} & 
\ct{2, 2}{\smcite{ABB+15}} & 
\cc{2, 3}{\smcite{ABB+15}} &
\ct{2, 2}{$\wedge\circ\vee$} & 
\cc{2, 3}{$\wedge\circ\vee$} & 
\cc{2, 3}{$\wedge\circ\vee$} &
\co{3,\Huangbound{6}}{\smcite{BHT17}} &
\co{4,\Huangbound{6}}{\smcite{ABB+15}} &
\co{2, 3}{\smcite{ABB+15}} & 
\co{2, 3}{\smcite{GPW15}} & 
\ct{4,\newbound{4}}{\smcite{ABB+15}} &
\cc{4, 6}{\smcite{ABB+15}} \\ \hline

$\R_0$ &
\ct{1, 1}{$\oplus$} &
\cellcolor{darkgray} & 
\ct{2, 2}{\smcite{ABB+15}} &
\ct{2, 2}{$\wedge\circ\vee$} & 
\cc{2, 3}{$\wedge\circ\vee$} & 
\cc{2, 3}{$\wedge\circ\vee$} &
\co{3,\Huangbound{6}}{\smcite{BHT17}} &
\co{3,\Huangbound{6}}{\smcite{BHT17}} &
\co{2, 3}{\smcite{ABB+15}} &
\co{2, 3}{\smcite{GJPW15}} &
\co{3,\newbound{4}}{\smcite{ABB+15}} & 
\cc{4, 6}{\smcite{ABB+15}} \\ \hline

$\R$ &
\ct{1, 1}{$\oplus$} & 
\ct{1, 1}{$\oplus$} & 
\cellcolor{darkgray} & 
\ct{2, 2}{$\wedge\circ\vee$} &
\cc{2, 3}{$\wedge\circ\vee$} & 
\cc{2, 3}{$\wedge\circ\vee$} &
\co{3,\Huangbound{6}}{\smcite{BHT17}} &
\co{3,\Huangbound{6}}{\smcite{BHT17}} &
\co{$\frac{3}{2}$, 3}{\smcite{ABB+15}} &
\co{2, 3}{\smcite{GJPW15}} &
\co{$\frac{8}{3}$,\newbound{4}}{\smcite{Tal19}} & 
\cc{4, 6}{\smcite{ABB+15}} \\ \hline

$\C$ &
\ct{1, 1}{$\oplus$} & 
\ct{1, 1}{$\oplus$} & 
\co{1, 2}{$\oplus$} & 
\cellcolor{darkgray} &
\ct{2, 2}{\smcite{GSS13}} & 
\ct{2, 2}{\smcite{GSS13}} &
\co{2.22,\Huangbound{5}}{\smcite{BHT17}} &
\co{2.22,\Huangbound{6}}{\smcite{BHT17}} &
\co{1.15, 3}{\smcite{Amb13}}  & 
\co{1.63, 3}{\smcite{NW95}} & 
\co{2, 4}{$\wedge$} & 
\co{2, 4}{$\wedge$} \\ \hline

$\RC$ &
\ct{1, 1}{$\oplus$} & 
\ct{1, 1}{$\oplus$} & 
\ct{1, 1}{$\oplus$} & 
\ct{1, 1}{$\oplus$} & 
\cellcolor{darkgray} & 
\co{$\frac{3}{2}$, 2}{\smcite{GSS13}} &
\co{2,\Huangbound{4}}{\smcite{Rub95}} &
\co{2,\Huangbound{4}}{$\wedge$} &
\co{1.15, 2}{\smcite{Amb13}} & 
\co{1.63, 2}{\smcite{NW95}} & 
\ct{2, 2}{$\wedge$} & 
\ct{2, 2}{$\wedge$} \\ \hline

$\bs$ &
\ct{1, 1}{$\oplus$} & 
\ct{1, 1}{$\oplus$} & 
\ct{1, 1}{$\oplus$} & 
\ct{1, 1}{$\oplus$} & 
\ct{1, 1}{$\oplus$} & 
\cellcolor{darkgray} &
\co{2,\Huangbound{4}}{\smcite{Rub95}} &
\co{2,\Huangbound{4}}{$\wedge$} &
\co{1.15, 2}{\smcite{Amb13}} & 
\co{1.63, 2}{\smcite{NW95}} & 
\ct{2, 2}{$\wedge$} & 
\ct{2, 2}{$\wedge$} \\ \hline

$\s$ &
\ct{1, 1}{$\oplus$} &
\ct{1, 1}{$\oplus$} &
\ct{1, 1}{$\oplus$} &
\ct{1, 1}{$\oplus$} &
\ct{1, 1}{$\oplus$} &
\ct{1, 1}{$\oplus$} &
\cellcolor{darkgray} &
\ct{2,\newbound{2}}{$\wedge$} &
\co{1.15, 2}{\smcite{Amb13}} &
\co{1.63, 2}{\smcite{NW95}}  &
\ct{2, 2}{$\wedge$} & 
\ct{2, 2}{$\wedge$} \\ \hline

$\lambda$ &
\ct{1, 1}{$\oplus$} &
\ct{1, 1}{$\oplus$} &
\ct{1, 1}{$\oplus$} &
\ct{1, 1}{$\oplus$} &
\ct{1, 1}{$\oplus$} &
\ct{1, 1}{$\oplus$} &
\ct{1, 1}{$\oplus$} &
\cellcolor{darkgray} &
\ct{1, 1}{$\oplus$} &
\co{1, 2}{$\oplus$}  &
\ct{1, 1}{$\oplus$} &
\co{1, 2}{$\oplus$}  \\ \hline

$\Q_E$ &
\ct{1, 1}{$\oplus$} &
\co{1.33, 2}{$\bar{\wedge}$-tree} &
\co{1.33, 3}{$\bar{\wedge}$-tree} &
\ct{2, 2}{$\wedge\circ\vee$} &
\cc{2, 3}{$\wedge\circ\vee$} &
\cc{2, 3}{$\wedge\circ\vee$} &
\co{3,\Huangbound{6}}{\smcite{BHT17}} &
\co{3,\Huangbound{6}}{\smcite{BHT17}} &
\cellcolor{darkgray} &
\co{2, 3}{\smcite{ABK16}} &
\co{2,\newbound{4}}{$\wedge$} &
\cc{4, 6}{\smcite{ABK16}} \\ \hline

$\deg$ & 
\ct{1, 1}{$\oplus$} &
\co{1.33, 2}{$\bar{\wedge}$-tree} &
\co{1.33,\Huangbound{2}}{$\bar{\wedge}$-tree} &
\ct{2, 2}{$\wedge\circ\vee$} &
\ct{2,\Huangbound{2}}{$\wedge\circ\vee$} &
\ct{2,\Huangbound{2}}{$\wedge\circ\vee$} &
\ct{2,\Huangbound{2}}{$\wedge\circ\vee$} &
\ct{2,\Huangbound{2}}{$\wedge$} &
\ct{1, 1}{$\oplus$}&
\cellcolor{darkgray} &
\ct{2,\newbound{2}}{$\wedge$} &
\co{2,\Huangbound{4}}{$\wedge$} \\ \hline

$\Q$ &
\ct{1, 1}{$\oplus$} &
\ct{1, 1}{$\oplus$} &
\ct{1, 1}{$\oplus$} &
\ct{2, 2}{\smcite{ABK16}} &
\cc{2, 3}{\smcite{ABK16}} &
\cc{2, 3}{\smcite{ABK16}} &
\co{3,\Huangbound{6}}{\smcite{BHT17}} &
\co{3,\Huangbound{6}}{\smcite{BHT17}} &
\ct{1, 1}{$\oplus$} &
\co{2, 3}{\smcite{ABK16}} &
\cellcolor{darkgray} &
\cc{4, 6}{\smcite{ABK16}} \\ \hline

\raisebox{-2pt}{$\adeg$} &
\ct{1, 1}{$\oplus$} &
\ct{1, 1}{$\oplus$} &
\ct{1, 1}{$\oplus$} &
\ct{$2$, 2}{\smcite{BT17}} &
\ct{$2$,\Huangbound{2}}{\smcite{BT17}} &
\ct{$2$,\Huangbound{2}}{\smcite{BT17}} &
\ct{2,\Huangbound{2}}{\smcite{BT17}} &
\ct{2,\Huangbound{2}}{\smcite{BT17}} &
\ct{1, 1}{$\oplus$} &
\ct{1, 1}{$\oplus$}  &
\ct{1, 1}{$\oplus$} &
\cellcolor{darkgray} \\ 
\end{tabular}
\label{tab:sep}
\end{center}
\begin{itemize}[itemsep=3pt]
    \item  An entry $a,b$ in the row $M_1$ and column $M_2$ roughly means that for all total functions $f$, $M_1(f) \leq M_2(f)^{b+o(1)}$ and there exists a function $g$ with $M_1(g) \geq M_2(g)^{a-o(1)}$ (see \cite{ABK16} for a precise definition). 
    \item {The second row of each cell contains an example of a function that achieves the separation (or a citation to an example), 
    where $\oplus = \textsc{parity}$, 
    $\wedge = \textsc{and}$, 
    $\vee = \textsc{or}$, 
    $\wedge \circ \vee = \textsc{and-or}$, 
    and $\bar{\wedge}$-tree is the balanced \textsc{nand}-tree function.}
    \item Cells have a white background if the relationship is optimal and a gray background otherwise.
    \item Entries with a \Huangbound{green} background follow from Huang's result. Entries with a \newbound{red} background follow from this work.
\end{itemize}
\end{minipage}
\end{table}
\renewcommand{\arraystretch}{1}

\subsection{Paper organization}

\Cref{sec:prelim} contains some preliminaries required to understand the proof of \Cref{thm:deg vs. Q}, which is proved in \Cref{sec:mainproof}. \Cref{sec:AKR}  gives some background and motivation for the Aanderaa--Karp--Rosenberg conjecture and proves \Cref{thm:qAKR}. We end with some open problems in \Cref{sec:open}. 

\Cref{sec:properties} describes some properties of $\lambda(f)$, its many equivalent formulations, and its relationship with other complexity measures.

\section{Preliminaries}\label{sec:prelim}

\subsection{Query complexity}

Let $f:\B^n \to\B$ be a Boolean function.
Let $A$ be a deterministic
algorithm that computes $f(x)$ on input
$x\in\B^n$ by making queries to the bits of $x$.
The worst-case number of queries $A$ makes (over choices
of $x$) is the query complexity of $A$. The minimum query
complexity of any deterministic algorithm computing $f$
is the deterministic query complexity of $f$, denoted
by $\D(f)$.

We define the bounded-error
randomized (respectively quantum) query complexity
of $f$, denoted by $\R(f)$ (respectively $\Q(f)$),
in an analogous way. We say an algorithm $A$ computes $f$
with bounded error if $\Pr[A(x)=f(x)]\geq 2/3$ for all
$x\in\B^n$, where the probability is over the internal
randomness of $A$. Then $\R(f)$ (respectively $\Q(f)$)
is the minimum number of queries required by any
randomized (respectively quantum) algorithm that computes $f$
with bounded error. It is clear that $\Q(f)\leq \R(f)\leq \D(f)$.
For more details on these measures,
see the survey by Buhrman and de Wolf \cite{BdW02}.

\subsection{Sensitivity and block sensitivity}

Let $f:\B^n\to\B$ be a Boolean function, and let $x\in\B^n$
be a string. A block is a subset of $[n]$. We say that a block $B\in[n]$
is sensitive for $x$ (with respect to $f$) if
$f(x \oplus \mathbbold{1}_B)\neq f(x)$, where $\mathbbold{1}_B$ is the $n$-bit string that
is $1$ on bits in $B$ and $0$ otherwise. 
We say a bit $i$ is sensitive for $x$
if the block $\{i\}$ is sensitive for $x$.
The maximum number of disjoint blocks
that are all sensitive for $x$
is called the block sensitivity of $x$ (with respect to $f$),
denoted by $\bs_x(f)$. The number of sensitive bits for $x$
is called the sensitivity of $x$,
denoted by $\s_x(f)$. Clearly, $\bs_x(f)\geq\s_x(f)$,
since $\s_x(f)$ is has the same definition as $\bs_x(f)$
except that the size of the blocks is restricted to $1$.
We define $\s(f) = \max_{x\in \B^n}{\s_x(f)}$ and $\bs(f) = \max_{x\in \B^n}{\bs_x(f)}$.

\subsection{Degree measures}\label{sec:degs}

A polynomial $q \in \Reals[x_1, \ldots, x_n]$ is said to represent
the function $f:\B^n\to\B$ if
$q(x)=f(x)$ for all $x\in\B^n$.
A polynomial $q$ is said to $\eps$-approximate $f$ if $q(x)\in[0,\eps]$
for all $x\in f^{-1}(0)$ and $q(x)\in[1-\eps,1]$ for all
$x\in f^{-1}(1)$.
The degree of $f$, denoted by
$\deg(f)$, is the minimum degree of a polynomial representing $f$.
The $\eps$-approximate degree, denoted by $\tilde{\deg}_\eps(f)$,
is the minimum degree of a polynomial $\eps$-approximating $f$.
We will omit $\eps$ when $\eps=1/3$.
We know that $\D(f)\geq\deg(f)$,
$\R(f)\geq\tilde{\deg}(f)$, and $\Q(f)\geq\tilde{\deg}(f)/2$.

The degree of $f$ as a polynomial is also called the Fourier-degree of $f$, which equals $\max\{|S|:|\hat{f}(S)|\neq 0\}$ where $\hat{f}(S) := \E_x[f(x) \cdot (-1)^{\sum_{i\in S} x_i}]$. In particular, $\deg(f)<n$ if and only if $f$ agrees with the Parity function, $\Parity_n(x) = \oplus_{i=1}^{n}x_i$, on exactly half of the inputs.

\section{Proof of main result (\texorpdfstring{\Cref{thm:deg vs. Q}}{Theorem 2})}\label{sec:mainproof}

Before proving \Cref{thm:deg vs. Q}, which is based on Huang's proof, we reinterpret his result in terms of a new complexity measure of Boolean functions that we call $\lambda(f)$: the spectral norm of the sensitivity graph of $f$.

\begin{definition}[Sensitivity Graph $G_f$, Spectral Sensitivity $\lambda(f)$]
	Let $f:\B^n \to \B$ be a Boolean function.
	The {\sf sensitivity graph} of $f$, $G_f = (V,E)$ is a subgraph of the Boolean hypercube, where $V= \B^n$, and $E = \{(x,x\oplus e_i)\in V\times V: i\in [n], f(x)\neq f(x\oplus e_i)\}$. That is, $E$ is the set of edges between neighbors on the hypercube that have different $f$-value. Let $A_f$ be the adjacency matrix of the graph $G_f$.  We define the {\sf spectral sensitivity} of $f$ as $\lambda(f)=\norm{A_f}$.
\end{definition}

Note that because $A_f$ is a real symmetric matrix, $\lambda(f)$ is also the largest eigenvalue of $A_f$. Since $G_f$ is bipartite, the largest and smallest eigenvalues of $A_f$ are equal in magnitude. 

Huang's proof of the sensitivity conjecture can be divided into two steps:
\begin{enumerate}
	\item $\forall{f}: \deg(f) \le \lambda(f)^2$
	\item $\forall{f}: \lambda(f) \le \s(f)$
\end{enumerate} 

The second step is the simple fact that the spectral norm of an adjacency matrix is at most the maximum degree of any vertex in the graph, which equals $\s(f)$ in this case.

We reprove the first claim, i.e., $\deg(f) \le \lambda(f)^2$, for completeness.

\begin{theorem}[\cite{Huang2019}]\label{lemma:Huang}\label{thm:Huang}
For all  Boolean functions $f:\B^n \to \B$, we have
	$\deg(f) \le \lambda(f)^2$.
\end{theorem}
\begin{proof}
	Without loss of generality we can assume that $\deg(f) = n$ since otherwise we can restrict our attention to a subcube of dimension $\deg(f)$ in which the degree remains the same and the top eigenvalue is at most $\lambda(f)$. Specifically, we can choose any monomial in the polynomial representing $f$ of degree $\deg(f)$ and set all the variables not appearing in this monomial to $0$.
		
	For $f$ with $\deg(f)=n$, let $V_0 = \{x\in \B^n: f(x) = \Parity_n(x)\}$ and $V_1=\{x\in \B^n: f(x)\neq \Parity_n(x)\}$.
	By the fact that $\deg(f)=n$ we know that $|V_0|\neq |V_1|$ as otherwise $f$ would have $0$ correlation with the $n$-variate parity function, implying that $f$'s top Fourier coefficient is $0$.
	
	We also note that any edge in the hypercube that goes between $V_0$ and $V_0$ is an edge in $G_f$ since it changes the value of $f$. This holds since for such an edge, $(x,x\oplus e_i)$, we have $f(x) = \Parity_n(x)\neq \Parity_n(x \oplus e_i) = f(x \oplus e_i)$.
	Similarly, any edge in the hypercube that goes between $V_1$ and $V_1$ is an edge in $G_f$. 
	
Assume without loss of generality that $|V_0| > |V_1|$. Thus, $|V_0|\ge 2^{n-1} + 1$. We will show that there exists a nonzero vector $v'$ supported only on the entries of $V_0$, such that $\|A_f \cdot v'\| \ge \sqrt{n} \cdot \|v'\|$.

Let $G=(V,E)$ be the complete $n$-dimensional Boolean Hypercube. That is, $V = \B^n$ and $E = \{(x,x\oplus e_i)\;:\;x\in \B^n, i\in [n]\}$.
Take the following signing of the edges of the Boolean hypercube, defined recursively. 
\begin{equation}
B_1 = \begin{pmatrix}0&1 \\1&0\end{pmatrix} \  \mathrm{and} \  B_i = \begin{pmatrix}B_{i-1}&I \\I&-B_{i-1}\end{pmatrix} \   \text{for $i\in \{2,\ldots,n\}$}.    
\end{equation}

This gives a new matrix $B_n \in \{-1,0,1\}^{V\times V}$ where $B_n(x,y)=0$ if and only if $x$ is not a neighbor of $y$ in the hypercube. 

Huang showed that $B_n$ has $2^{n}/2$ eigenvalues that equal $-\sqrt{n}$ and $2^{n}/2$ eigenvalues that equal $+\sqrt{n}$.
To show this, he showed that $B_n^2 = n\cdot I$ by induction on $n$ and thus all eigenvalues of $B_n$ must be either $+\sqrt{n}$ or $-\sqrt{n}$. Then, observing that the trace of $B_n$ is $0$, as all diagonal entries equal $0$, we see that we must have an equal number of $+\sqrt{n}$ and $-\sqrt{n}$ eigenvalues. 

Thus, the subspace of eigenvectors for $B_n$ with eigenvalue $\sqrt{n}$ is of dimension $2^{n}/2$. Using $|V_1|<2^n/2$, there must exists a nonzero eigenvector for $B_n$ with eigenvalue $\sqrt{n}$ that vanishes on $V_1$. 
Fix $v$ to be any such vector. 

Let $v'$  be the vector whose entries are the absolute values of the entries of $v$.
We claim that $\|A_f \cdot v'\|_2 \ge \sqrt{n}\cdot \|v'\|_2$.
To see so, note that for every $x\in V_0$ we have 
\begin{align}(A_f \cdot v')_x &= \sum_{y\sim x:f(y)\neq f(x)} v'_y = \sum_{y\sim x:y\in V_0} v'_y = \sum_{y\sim x} v'_y \nonumber \\ 
	&\ge \sum_{y\in \B^n} |B_{x,y} v_y| \ge  \left|\sum_{y\in \B^n}B_{x,y} v_y\right| = \sqrt{n}\cdot |v_x| = \sqrt{n} \cdot v'_x\;.
\end{align}
On the other hand, for $x\in V_1$ we have $(A_f \cdot v')_x = 0 = v'_x$.
Thus the norm of $A_f\cdot v'$ is at least $\sqrt{n}$ times the norm of $v'$, and hence $\lambda(f)= \|A_f\| \ge \sqrt{n} = \sqrt{\deg(f)}$. 
\end{proof}

Finally, we prove that $\lambda(f) = O(\Q(f))$.
We rely on a variant of the adversary method introduced by Barnum, Saks, and Szegedy \cite{BSS03} (see also~\cite{SSpalekS06}).

\begin{definition}[Spectral Adversary method]
Let $\{D_i\}_{i\in [n]}$ and $F$ be matrices of size $\B^n \times \B^n$ with entries in $\B$ satisfying $D_i[x,y]=1$ if and only if $x_i\neq y_i$, and $F[x,y]=1$ if and only if $f(x)\neq f(y)$. Let $\Gamma$ denote a $\B^n \times \B^n$ nonnegative symmetric matrix such that $\Gamma\circ F = \Gamma$ (i.e., the nonzero entries of  $\Gamma$ are a subset of the  the nonzero entries of $F$). Then 
$\SA(f) = \max_{\Gamma} \frac{\|\Gamma\|}{\max_{i\in[n]}{\|\Gamma \circ D_i\|}}$.
\end{definition}
Barnum, Saks, and Szegedy \cite{BSS03} proved that $\Q(f) = \Omega(\SA(f))$.
\begin{lemma}\label{lemma:Q vs rho}
	For all Boolean functions $\Q(f) = \Omega(\SA(f)) = \Omega(\lambda(f))$.
\end{lemma}
\begin{proof}
	We prove that $\SA(f) \ge \lambda(f)$.
	Indeed, one can take $\Gamma$ to be simply the adjacency matrix of $G_f$. That is, for any $x,y\in \B^n$ put $\Gamma[x,y] = 1$ if and only if $y \sim x$ in the hypercube and $f(x)\neq f(y)$. We observe that $\|\Gamma\|=\lambda(f)$. On the other hand, for any $i\in [n]$, $\Gamma \circ D_i$ is  the restriction of the sensitive edges in direction $i$. 
	The maximum degree in the graph represented by $\Gamma \circ D_i$ is $1$ hence  $\|\Gamma \circ D_i\|$ is at most $1$.
	Thus we have 
	\begin{equation}
			\SA(f) \ge \frac{\|\Gamma\|}{\max_{i\in[n]}{\|\Gamma \circ D_i\|}} \ge \lambda(f).
	\end{equation}
	Combining this with $\Q(f) = \Omega(\SA(f))$~\cite{BSS03}, we get $\Q(f) = \Omega(\SA(f)) = \Omega(\lambda(f))$.
\end{proof}

From \Cref{lemma:Huang} and \Cref{lemma:Q vs rho} we immediately get \Cref{thm:deg vs. Q}.

\section{Monotone graph properties}
\label{sec:AKR}

The Aanderaa--Karp--Rosenberg conjectures are a collection of conjectures related to the query complexity of deciding whether an input graph specified by its adjacency matrix satisfies a given property in various models of computation. 

Specifically, let the input be an $n$-vertex undirected simple graph specified by its adjacency matrix. This means we can query any unordered pair $\{i,j\}$, where $i,j\in[n]$, and learn whether there is an edge between vertex $i$ and $j$. Note that the input size is $\binom{n}{2}=\Theta(n^2)$. 

A function $f$ on $\binom{n}{2}$ variables is a graph property if it treats the input as a graph and not merely a string of length $\binom{n}{2}$. Specifically, the function must be invariant under permuting vertices of the graph. 
In other words, the function can only depend on the isomorphism class of the graph, not the specific labels of the vertices. 
A function $f$ is monotone (increasing) if for all $x,y\in\B^n$, $x \leq y \implies f(x) \leq f(y)$, where $x \leq y$ means $x_i \leq y_i$ for all $i\in[n]$. 
For a monotone function, negating a $0$ in the input cannot change the function value from $1$ to $0$. 
In the context of graph properties, if the input graph has a certain monotone graph property, then adding more edges cannot destroy the property. 

Examples of monotone graph properties include ``$G$ is connected,'' ``$G$ contains a clique of size $k$,'' ``$G$ contains a Hamiltonian cycle,'' ``$G$ has chromatic number greater than $k$,'' ``$G$ is not planar'', and ``$G$ has diameter at most $k$.'' Many commonly encountered graph properties (or their negation) are monotone graph properties. Finally, we say a function $f:\B^n\to\B$ is nontrivial if there exist inputs $x$ and $y$ such that $f(x)\neq f(y)$.

The deterministic Aanderaa--Karp--Rosenberg conjecture, also called the \emph{evasiveness conjecture},\footnote{A function $f$ is called \emph{evasive} if its deterministic query complexity equals its input size.} states that for all nontrivial monotone graph properties $f$, $\D(f) = \binom{n}{2}$. This conjecture remains open to this day, although the weaker claim that $\D(f)=\Omega(n^2)$ was proved over 40 years ago by Rivest and Vuillemin~\cite{RV76}. Several works have improved on the constant in their lower bound, and the best current result is due to Scheidweiler and Triesch~\cite{ST13}, who prove a lower bound of $\D(f)\geq (1/3-o(1)) \cdot n^2$. The evasiveness conjecture has been established in several special cases including when $n$ is prime~\cite{KSS84} and when restricted to bipartite graphs~\cite{Yao88}.

The randomized Aanderaa--Karp--Rosenberg conjecture asserts that all nontrivial monotone graph properties $f$ satisfy $\R(f) = \Omega(n^2)$. A sequence of increasingly stronger lower bounds, starting with a lower bound of $\Omega(n\log^{1/12}n)$ due to Yao~\cite{Yao91}, a lower bound of $\Omega(n^{5/4})$ due to King~\cite{Kin88}, and a lower bound of $\Omega(n^{4/3})$ due to Hajnal~\cite{Haj91}, has led to the current best lower bound of $\Omega(n^{4/3}\log^{1/3}n)$ due to Chakrabarti and Khot~\cite{CK01}. There are also two lower bounds due to  Friedgut, Kahn, and Wigderson~\cite{FKW02} and O'Donnell, Saks,  Schramm, and Servedio~\cite{OSSS05} that are better than this bound for some graph properties.

The quantum Aanderaa--Karp--Rosenberg conjecture states that all nontrivial monotone graph properties $f$ satisfy $\Q(f) = \Omega(n)$. This is the best lower bound one could hope to prove since there exist properties with $\Q(f) = O(n)$, such as the property of containing at least one edge. In fact, for any $\alpha \in [1,2]$ it is possible to construct a graph property with quantum query complexity $\Theta(n^{\alpha})$ using known lower bounds for the threshold function~\cite{BBCMW01}.

As stated in the introduction, the question was first raised by Buhrman, Cleve, de Wolf, and Zalka~\cite{BCdWZ99}, who showed a lower bound of $\Omega(\sqrt{n})$. This was improved by Yao to $\Omega(n^{2/3}\log^{1/6}n)$ using the technique in \cite{CK01} and Ambainis' adversary bound~\cite{Amb02}. Better lower bounds are known in some special cases, such as when the property is a subgraph isomorphism property, where we know a lower bound of $\Omega(n^{3/4})$ due to Kulkarni and Podder~\cite{KP16}.

As stated in \Cref{thm:qAKR}, we resolve the quantum Aanderaa--Karp--Rosenberg conjecture and show an optimal $\Omega(n)$ lower bound. The proof combines \Cref{thm:deg vs. Q} with a quadratic lower bound on the degree of nontrivial monotone graph properties. With some work, the original quadratic lower bound on the deterministic query complexity of nontrivial monotone graph properties by Rivest and Vuillemin~\cite{RV76} can be modified to prove a similar lower bound for degree. 
We were not able to find such a proof in the literature, and instead combine the following two claims to obtain the desired claim.

First, we use the result of Dodis and Khanna~\cite[Theorem 2]{DK99}:
\begin{theorem}
For all nontrivial monotone graph properties, $\deg_2(f) = \Omega(n^2)$.
\end{theorem}

Here $\deg_2(f)$ is the minimum degree of a Boolean function when represented as a polynomial over the finite field with two elements, $\mathbb{F}_2$. We combine this with a standard lemma that shows that this measure lower bounds $\deg(f)$. A proof can be found in~\cite[Proposition 6.23]{ODo09}:
\begin{lemma}
For all Boolean functions $f:\B^n \to \B$, we have $\deg_2(f) \leq \deg(f)$.
\end{lemma}

Combining these with \Cref{thm:deg vs. Q}, we get that all nontrivial monotone graph properties $f$ satisfy $\Q(f) = \Omega(n)$, which is the statement of \Cref{thm:qAKR}.

\section{Open questions}\label{sec:open}

We saw that $\lambda(f)$ lower-bounds both $\Adv(f)$, and thus $\Q(f)$, and also the sensitivity $\s(f)$. 
One might conjecture that $\lambda(f)$ lower-bounds all the complexity measures in \Cref{fig:rel}, including $\tilde{\deg}(f)$.
\begin{conjecture}\label{conj:adeg}
	For all Boolean functions $f:\B^n \to \B$, we have $\lambda(f) = O(\tilde{\deg}(f))$.
\end{conjecture}

If \Cref{conj:adeg} we true, \Cref{lemma:Huang} would imply that $\deg(f) = O(\tilde{\deg}(f)^2)$, settling a longstanding conjecture posed by Nisan and Szegedy~\cite{NS94}. The current best relation between the two measures is $\deg(f) = O(\tilde{\deg}(f)^6)$.
The following conjecture is weaker, and might be easier to tackle first.
\begin{conjecture}\label{conj:deg}
	For all Boolean functions $f:\B^n \to \B$, we have $\lambda(f) = O(\deg(f))$.
\end{conjecture}

Another longstanding open problem is to show a quadratic relation between deterministic query complexity and block sensitivity:
\begin{conjecture}
    For all Boolean functions $f:\B^n \to \B$, we have $\D(f) = O(\bs(f)^2)$.
\end{conjecture}
If this conjecture were true, it would optimally resolve several relationships in \Cref{tab:sep}, and would imply, for example, $\D(f) = O(\R(f)^2))$ and $\D(f) = O(\adeg(f)^4)$.

After settling the best relation between $\D(f)$ and $\Q(f)$, the next pressing question is to settle the best relation between  $\R(f)$ and $\Q(f)$.
Recently, the fourth author~\cite{Tal19} showed a power $8/3$ separation between $\R(f)$ and $\Q(f)$, while the best known relationship is a power $4$ relationship (this work). We conjecture that both these bounds can be improved.

\begin{conjecture}
    For all Boolean functions $f:\B^n \to \B$, we have $\R(f) = O(\Q(f)^3)$.
\end{conjecture}

\begin{conjecture}\label{conj:cubic separation}
    There exists a Boolean function $f:\B^n \to \B$ such that $\R(f) = \Omega(\Q(f)^3)$.
\end{conjecture}
We note that there are candidate constructions based on the work of~\cite{AA18, ABK16, Tal19} that are conjectured to satisfy $\Q(f) \ge  \R(f)^{3-o(1)}$. In particular, it suffices to prove a conjectured bound on the Fourier spectrum of deterministic decision trees \cite{Tal19} to prove \cref{conj:cubic separation}.

Finally, for the special case of {\sf monotone} total Boolean functions $f$, Beals et al.~\cite{BBCMW01} already showed in 1998 that $\D(f)=O(\Q(f)^4)$. It would be interesting to know whether this can be improved, perhaps all the way to $\D(f)=O(\Q(f)^2)$.

\bibliographystyle{alphaurl}
\bibliography{bibs}

\appendix

\section{Properties of the measure \texorpdfstring{$\lambda(f)$}{lambda}}
\label{sec:properties}

We show that the measure $\lambda(f)$ satisfies various elegant
properties. First, it can be defined in multiple ways,
one of which was introduced by Koutsoupias
back in 1993 \cite{Kou93}. It also has a formulation
as a special case of the quantum adversary bound and hence can 
be expressed as as a semidefinite program closely related to that
of the quantum adversary bound. Due to this characterization, $\lambda(f)$
can be viewed as both a maximization problem and a minimization
problem. These equivalent formulations are described in \Cref{sec:equivalent}.

Second,
we show that $\lambda(f)\le\sqrt{\s_0(f)\s_1(f)}$, 
which was already observed by Laplante, Lee, and Szegedy~\cite{LLS06}
(though we give a slightly different proof).
Finally, we show lower bounds on $\lambda(f)$ and an optimal quadratic separation between $\lambda(f)$ and $\s(f)$.

\subsection{Equivalent formulations}\label{sec:equivalent}

\begin{theorem}\label{thm:equivalent}
For all Boolean functions $f\colon\B^n\to\B$, we have
\begin{equation}
\lambda(f)=\K(f)=\Adv_1(f)=\Adv_1^{\pm}(f),    
\end{equation}
where the measures $\K(f)$, $\Adv_1(f)$, and $\Adv_1^{\pm}(f)$
are defined below. Furthermore, $\Adv_1(f)$ itself
has several equivalent formulations:
$\Adv_1(f)\coloneqq\SA_1(f)=\SWA_1(f)=\MM_1(f)=\GSA_1(f)$.
\end{theorem}

We now define all these measures before proving this theorem.

\paragraph{Koutsoupias complexity $\K(f)$.}
For a Boolean function $f$, let $A\subseteq f^{-1}(0)$,
and let $B\subseteq f^{-1}(1)$. Let $Q$ be the matrix
with rows and columns labeled by $A$ and $B$ respectively,
with $Q[x,y]=1$ if the Hamming distance of $x$ and $y$ is $1$,
and $Q[x,y]=0$ otherwise. Koutsoupias \cite{Kou93}
observed that $\|Q\|^2$ is a lower bound on formula size,
for every such choice of $A$ and $B$. We define
$\K(f)$ to be the maximum value of $\|Q\|$ over choices
of $A$ and $B$.
Thus $\K(f)^2$ is a lower bound on the formula size of $f$.

\paragraph{Single-bit positive adversary $\Adv_1(f)$.}
We define $\Adv_1(f)$ as a version of the adversary bound
where we are only allowed to put nonzero weight on input pairs
$(x,y)$ where $f(x)\ne f(y)$ and the Hamming distance between
$x$ and $y$ is exactly $1$. We will define $\Adv_1(f)$
in terms of the spectral adversary version, which we
also denote by $\SA_1(f)$.
$\Adv_1(f) = \SA_1(f)$ is defined as the maximum of
\begin{equation}
\frac{\|\Gamma\|}{\max_{i\in[n]}\|\Gamma\circ D_i\|}    
\end{equation}
over matrices $\Gamma$ of a special form. We require $\Gamma$
satisfy the following: (1) its entries are nonnegative reals;
(2) its rows and columns are indexed by $\Dom(f)$;
(3) $\Gamma[x,y]=0$ whenever $f(x)=f(y)$; 
(4) $\Gamma[x,y]=0$ whenever the Hamming distance of $x$ and $y$
is not $1$; and (5) $\Gamma$ is not all $0$.
In the above expression, $\circ$ refers to the Hadamard
(entrywise) product, $\Dom(f)$ is the domain of $f$, 
and $D_i$ is the $\B$-valued matrix with $D_i[x,y]=1$
if and only if $x_i\ne y_i$.

\paragraph{Single-bit negative adversary $\Adv_1^{\pm}(f)$.}
We define $\Adv_1^{\pm}(f)$ using the same definition as
$\Adv_1(f)$ above, except that the matrix $\Gamma$ is allowed
to have negative entries. Note that since
this is a relaxation of the
conditions on $\Gamma$, we clearly have
$\Adv_1^\pm(f)\ge\Adv_1(f)$.

\paragraph{Single-bit strong weighted adversary $\SWA_1(f)$.}
We define $\SWA_1(f)$ as a single-bit version of the
strong weighted adversary method $\SWA(f)$ from \cite{SSpalekS06}.
For this definition, we say a weight function
$w\colon\Dom(f)\times\Dom(f)\to[0,\infty)$ is feasible
if it is symmetric (i.e., $w(x,y)=w(y,x)$) and
if it satisfies the conditions on $\Gamma$ above
(i.e., it places weight $0$ on a pair $(x,y)$
unless both $f(x)\ne f(y)$ and the Hamming distance
between $x$ and $y$ is $1$).
We view such a feasible weight scheme $w$ as the weights
on a weighted bipartite graph, where the left vertex set is
$f^{-1}(0)$ and the right vertex set is $f^{-1}(1)$.
We let $wt(x)\coloneqq\sum_y w(x,y)$ denote the weighted degree of
$x$ in this graph, i.e., the sum of the weights of its incident
edges. Then $\SWA_1(f)$ is defined as the maximum, over
such feasible weight schemes $w$, of
\begin{equation}
\min_{x,i:w(x,x^i)>0}\frac{\sqrt{wt(x)wt(x^i)}}{w(x,x^i)}.    
\end{equation}
Here $x$ ranges over $\Dom(f)$, $i$ ranges over $[n]$, and
$x^i$ denotes the string $x$ with bit $i$ flipped.\footnote{%
Readers familiar with the adversary bound should note that
this definition is analogous a weighted version
of Ambainis's original adversary method; in the original method,
the denominator was the geometric mean of (a) the weight
of the neighbors of $x$ with disagree with $x$ at $i$,
and (b) the weight of the neighbors of $x^i$ which disagree
with $x^i$ at $i$; but in our case, both (a) and (b) are simply
$w(x,x^i)$, since $x^i$ is the only string that disagrees
with $x$ on bit $i$ and is connected to $x$ in the
bipartite graph.}

\paragraph{Single-bit minimax adversary $\MM_1(f)$.}
Unlike the other forms, we define $\MM_1(f)$ as a minimization
problem rather than a maximization problem. We
say a weight function $w\colon\Dom(f)\times[n]\to[0,\infty)$
is feasible if for all $x,y\in\Dom(f)$ with
$f(x)\ne f(y)$ and Hamming distance $1$, we have
$w(x,i)w(y,i)\ge 1$, where $i$ is the bit on which
$x$ and $y$ disagree. $\MM_1(f)$ is defined as the minimum,
over such feasible weight schemes $w$, of
\begin{equation}
    \max_{x\in\Dom(f)}\sum_{i\in[n]} w(x,i).
\end{equation}

\paragraph{Semidefinite program version $\GSA_1(f)$.}
We define $\GSA_1(f)$ to be the optimal value of the following
semidefinite program.
\begin{equation}
\begin{array}{lll}
\text{maximize}  &  \langle Z,A_f\rangle &\\
\text{subject to}& \Delta\mbox{ is diagonal} & \\
                & \tr\Delta = 1 & \\
                & \Delta-Z\circ D_i\succeq 0 & \forall i\in[n]\\
                & Z \ge 0 &
\end{array}
\end{equation}
Here $Z$ and $\Delta$ are variable matrices with rows
and columns indexed by $\Dom(f)$, $A_f$ is the $\B$-matrix with
$A_f[x,y]=1$ if and only if both $f(x)\ne f(y)$
and $(x,y)$ have Hamming distance $1$, and $D_i$
is the $\B$-matrix with $D_i[x,i]=1$ if and only if $x_i\ne y_i$.

We now prove \Cref{thm:equivalent}.

\begin{proof}
Recall that in the definition of $\K(f)$,
we picked $A\subseteq f^{-1}(0)$ and $B\subseteq f^{-1}(1)$
and defined the resulting matrix $Q$. Since the spectral norm
of a submatrix is always smaller than or equal to the spectral
norm of the original matrix, we can always assume without
loss of generality that $A=f^{-1}(0)$ and $B=f^{-1}(1)$.
Then $\K(f)=\|Q\|$ for the resulting matrix $Q$ with rows
and columns indexed by $f^{-1}(1)$ and $f^{-1}(0)$ respectively.
Now, recall that $A_f$ was the adjacency matrix of the graph
$G_f$, which has an edge between $x$ and $y$ if $f(x)\ne f(y)$
and the Hamming weight between $x$ and $y$ is $1$.
The rows and columns of $A_f$ are each indexed by $\Dom(f)$.
By rearranging them, we can make $A_f$ be block diagonal
with blocks equal to $Q$ and $Q^\dagger$. From there it
follows that $\|A_f\|=\|Q\|$, so $\lambda(f)=\K(f)$.

Next, recall that $\Adv_1(f)$ is defined as the maximum
ratio $\|\Gamma\|/\max_i\|\Gamma\circ D_i\|$ over valid
choices of $\Gamma$. Note that since $\Gamma[x,y]$
can only be nonzero if $x$ and $y$ disagree on one bit,
$\Gamma\circ D_i$ is nonzero only on pairs $(x,y)$
which disagree exactly on bit $i$. In other words,
if $P_i$ denotes the $\B$-valued matrix with $P_i[x,y]=1$
if and only if $x$ and $y$ disagree on bit $i$ and only on $i$,
then $\Gamma\circ D_i$ is nonzero only in entries where $P_i$
is $1$. Now, note that $P_i$ is a permutation matrix.
Hence, by rearranging the rows and columns of $\Gamma\circ D_i$,
we can get it to be diagonal. This means $\|\Gamma\circ D_i\|$
is the maximum entry of $\Gamma\circ D_i$, and hence
$\max_i\|\Gamma\circ D_i\|$ is the maximum entry of $\Gamma$.
It follows that $\Adv_1(f)$ is the maximum of $\|\Gamma\|$
over feasible matrices $\Gamma$ with $\max(\Gamma)\le 1$,
where $\max(\Gamma)=\max_{ij}|\Gamma_{ij}|$. This argument also holds for
$\Adv_1^\pm(f)$, which is the maximum of $\|\Gamma\|$
over feasible (possibly negative) matrices $\Gamma$
with $\max(\Gamma)\le 1$.

Next, observe that negative weights never help for maximizing
$\|\Gamma\|$: indeed, if we had $\Gamma$ with negative entries
maximizing $\|\Gamma\|$, then we would have vectors $u$ and $v$
with $\|u\|_2=\|v\|_2=1$ and $u^T\Gamma v=\|\Gamma\|$;
but then replacing $u$ and $v$ with their entry-wise absolute
values, and replacing $\Gamma$ with its entry-wise absolute
value $\Gamma'$, we clearly get that $\|\Gamma'\|\ge\|\Gamma\|$.
However, $\max(\Gamma')=\max(\Gamma)$, so $\Gamma'$
remains feasible. This means we can always take the maximizing
matrix $\Gamma$ to be nonnegative, so $\Adv_1^\pm(f)=\Adv_1(f)$.
We can similarly assume that the unit vectors $u$ and $v$
maximizing $u^T\Gamma v$ are nonnegative.

Finally, consider the maximizing matrix $\Gamma$ and the
maximizing unit vectors $u$ and $v$, all nonnegative,
and satisfying $\max(\Gamma)\le 1$. Note that
the expression $u^T\Gamma v$ is nondecreasing in the entries
of $\Gamma$, since everything is nonnegative. Hence
to maximize $u^T\Gamma v$, we can always take every nonzero
entry of $\Gamma$ to be $1$, since this maintains
$\max(\Gamma)\le 1$. In other words, the matrix maximizing
$\|\Gamma\|$ will always simply be $A_f$, and hence
$\Adv_1(f)$ is always exactly equal to $\lambda(f)$.

It remains to show that $\SA_1(f)=\SWA_1(f)=\MM_1(f)=\GSA_1(f)$.
The proof of this essentially
follows the arguments in \cite{SSpalekS06} for the regular
positive adversary, though some steps are a little simpler.
To start, we've seen that $\SA_1(f)=\lambda(f)$. Since
$A_f$ is symmetric, we have $\lambda(f)=v^T A_f v$ for
some unit vector $v$, which we've established is nonnegative;
this vector is also an eigenvector, so $A_f v=\lambda(f)v$.
Consider the weight scheme $w(x,y)=v[x]v[y]A_f[x,y]$. Then
$wt(x)=\sum_y v[x]v[y]A_f[x,y]=v[x](A_f v)[x]=\lambda(f)v[x]^2$.
Hence if $w(x,x^i)>0$, we have
\begin{equation}
    \frac{\sqrt{wt(x)wt(x^i)}}{w(x,x^i)}=\frac{\lambda(f)v[x]v[x^i]}{v[x]v[x^i]A_f[x,x^i]}=\lambda(f).
\end{equation}
This means $\SWA_1(f)\ge\SA_1(f)$.
In the other direction, let $w$ be a feasible weight scheme
for $\SWA_1(f)$, let $\Gamma[x,y]=w(x,y)/\sqrt{wt(x)wt(y)}$,
and let $v[x]=\sqrt{wt(x)/W}$, where $W=\sum_x wt(x)$.
Then $\|v\|_2^2=\sum_x wt(x)/W=1$, and
\begin{equation}
    v^T\Gamma v
=\sum_{x,y} \sqrt{wt(x)wt(y)}w(x,y)/W\sqrt{wt(x)wt(y)}
=(1/W)\sum_{x,y}w(x,y)=1.
\end{equation}
Hence $\|\Gamma\|\ge 1$. On the other hand, we have
$\max(\Gamma)=\max_{x,y} w(x,y)/\sqrt{wt(x)wt(y)}$.
This means that the ratio $\|\Gamma\|/\max(\Gamma)$
equals $\min_{x,y:w(x,y)>0}\sqrt{wt(x)wt(y)}/w(x,y)$,
which is $\SWA_1(f)$; thus $\SA_1(f)\ge\SWA_1(f)$.

Next we examine $\GSA_1(f)$. Consider a solution
$(Z,\Delta)$ to this semidefinite
program and define $\Gamma=Z\circ M\circ A_f$,
where $M$ is defined as $M=uu^T$ and $u$ is defined by
$u[x]=1/\sqrt{\Delta[x,x]}$ when $\Delta[x,x]>0$
and $u[x]=0$ otherwise. Recall that $\Delta$ is diagonal
and that $\Delta-Z\circ D_i\succeq 0$ for all $i$.
Since positive semidefinite matrices are symmetric,
$Z\circ D_i$ must be symmetric for all $i$, so $Z$
is symmetric. Moreover, the diagonal of $Z\circ D_i$
is all zeros, so we must have $\Delta\ge 0$.
Further, if $\Delta[x,x]=0$ for some $x$, we must have
the corresponding row and column of $Z$ be all zeros.
If we let $\Delta'$ and $Z'$ be $\Delta$ and $Z$ with the
all-zero rows and columns deleted, then it is clear that
$\Delta-Z\circ D_i\succeq 0$ if and only if
$\Delta'-Z'\circ D_i\succeq 0$. Defining $M'$ as $M$
with those rows and columns deleted and $u'$ as $u$ with
those entries deleted, we have $M'=u'(u')^T>0$.
Observe that $\Delta'-Z'\circ D_i\succeq 0$
if and only if $v^T(\Delta'-Z'\circ D_i)v\ge 0$ for all
vectors $v$, which is if and only if
$(v\circ u')^T(\Delta'-Z'\circ D_i)(v\circ u')\ge 0$
for all vectors $v$ (since we have $u'>0$). This, in turn,
is equivalent to $M'\circ (\Delta'-Z'\circ D_i)\succeq 0$.
Since $M'\circ \Delta'=I$, this is equivalent to
$I-M'\circ Z'\circ D_i\succeq 0$,
which is in turn equivalent to $I-M\circ Z\circ D_i\succeq 0$.
Since $Z\ge 0$ and we are maximizing $\langle Z,A_f\rangle$,
it never helps for $Z$ to have nonzero entries in places
where $A_f$ is $0$. Hence we can assume without loss of generality
that $Z=Z\circ A_f$, which means the constraint becomes
$I-\Gamma\circ D_i\succeq 0$, where we defined
$\Gamma=M\circ Z\circ A_f$. We thus have
$\|\Gamma\circ D_i\|\le 1$. On the other hand, letting
$v[x]=\sqrt{\Delta[x,x]}$, we have
\begin{equation}
    v^T\Gamma v=\sum_{x,y}v[x]v[y]M[x,y]Z[x,y]A_f[x,y]
=\sum_{x,y:\Delta[x,x],\Delta[y,y]>0}Z[x,y]A_f[x,y]
=\langle Z,A_f\rangle.
\end{equation}
Hence $\SA_1(f)\ge\GSA_1(f)$.
The reduction in the other direction works similarly:
start with an adversary matrix $\Gamma$ with
$\max(\Gamma)\le 1$, and let $v$ be its
principle eigenvector. Then set $Z=\Gamma\circ (vv^T)$
and $\Delta=I\circ (vv^T)$. Then $I-\Gamma\circ D_i\succeq 0$,
which implies that $\Delta-Z\circ D_i\succeq 0$.
We also have $\tr\Delta=1$, $Z\ge 0$, and
$\langle Z,A_f\rangle=\|\Gamma\|$.

Finally, we handle $\MM_1(f)$. To do so, we first take the
dual of the semidefinite program for $\GSA_1(f)$.
This dual has the form
\begin{equation}
\begin{array}{lll}
\text{minimize}  &  \alpha &\\
\text{subject to}& \sum_i R_i\circ I\le\alpha I & \\
                & \sum_i R_i\circ D_i\ge A_f & \\
                & R_i\succeq 0 & \forall i\in[n]
\end{array}
\end{equation}
where the variables are $\alpha$ (a scalar) and matrices
$R_i$, each with rows and columns indexed by $\Dom(f)$.
Strong duality follows since when
$A_f$ is not all zeros, and the semidefinite program
in $\GSA_1(f)$ has a strictly feasible solution
(just take $Z$ to equal $\epsilon A_f$ for a small enough
positive constant $\epsilon$, and take $\Delta=I/|\Dom(f)|$).
This means the optimal solution of the minimization
problem above equals $\Adv_1(f)$. It remains
to show that this optimal solution $T$ also equals
$\MM_1(f)$.

Let $\alpha$ and $\{R_i\}_i$ be a feasible solution to
the semidefinite program above. Since $R_i\succeq 0$,
we have $R_i=X_iX_i^T$ for some matrix $X_i$.
Define $w(x,i)=R_i[x,x]$.
Note that we also have $w(x,i)=\sum_a X_i[x,a]^2$.
Then by Cauchy--Schwarz,
$w(x,i)w(y,i)\ge
\left(\sum_a X_i[x,a]X_i[y,a]\right)^2
=(X_iX_i^T)[x,y]^2=R_i[x,y]^2$.
If $x$ and $y$ are such that $A_f[x,y]=1$, then they
disagree in only one bit $i$, and hence $D_i[x,y]=1$
for that $i$ and $D_j[x,y]=0$ for all $j\ne i$.
Since we have $\sum_i R_i\circ D_i\ge A_f$,
we conclude that for all such pairs $(x,y)$,
we have $w(x,i)w(y,i)\ge R_i[x,y]^2\ge A_f[x,y]^2=1$
on the bit $i$ where $x$ and $y$ differ; hence
the weight scheme $w$ is feasible.
Furthermore, for any $x$,
$\sum_i w(x,i)=\sum_i R_i[x,x]\le\alpha I[x,x]=\alpha$.
Hence $\MM_1(f)$ is at most the optimal value of this
semidefinite program.

In the other direction, consider a feasible weight scheme
$w$, and define $R_i[x,y]=\sqrt{w(x,i)w(y,i)}$.
Then $R_i=w(\cdot,i)w(\cdot,i)^T$, where we treat
$w(\cdot,i)$ as a vector; hence $R_i\succeq 0$.
Moreover, $R_i\ge 0$, and for a pair $(x,y)$
with $A_f[x,y]=1$, there is some $i$ which is the unique
bit they disagree on, and hence $w(x,i)w(y,i)\ge 1$;
but this means that $R_i[x,y]\ge 1$, and so
$(R_i\cdot D_i)[x,y]\ge 1=A_f[x,y]$.
Finally, $\sum_i R_i[x,x]=\sum_i w(x,i)$,
which means that $\sum_i R_i\circ I\le \MM_1(f)\cdot I$, as desired.
\end{proof}

\subsection{Upper bounds}

We now show a slightly better upper bound on $\lambda(f)$, that it is upper bounded by the geometric mean of the 0-sensitivity and 1-sensitivity, which can be a better upper bound than $\s(f)$. 

We provide two proofs of this. The first uses the $\lambda(f)$ formulation and uses a linear algebra argument about norms. This proof is due to Laplante, Lee, and Szegedy~\cite{LLS06}, who observed this about the measure $\K(f)$. 

To describe this proof, we briefly need to describe some matrix norms.
For a vector $v \in \Reals^n$, the $p$-norm for a positive integer $p$ is defined as $\norm{v}_p=(\sum_{i\in[n]}|v_i|^p)^{1/p}$. We also define $\norm{v}_\infty = \max_{i \in [n]} |v_i|$.
Note that $\norm{v}_1$ is simply the sum of the absolute values of all the entries of the vector.

Similarly, for a matrix $A \in \Reals^{n\times m}$, we define the induced $p$-norm of $A$ to be
\begin{equation}
    \norm{A}_p = \max\{\norm{Ax}_p: \norm{x}_p =1\}.
\end{equation}
The spectral norm $\norm{A}$ is the induced $2$-norm $\norm{A}_2$.
The 1-norm $\norm{A}_1$ is simply the maximum sum of absolute values of entries in any column of the matrix. The $\infty$-norm $\norm{A}_\infty$ is  the maximum sum of absolute values of entries in any row of the matrix. 

Lastly, we need a useful relationship between these norms sometimes called  H\"{o}lder's inequality for induced matrix norms (see \cite[Corollary 2.3.2]{GV13} for a proof): 
\begin{proposition}\label{prop:Holder}
For all matrices $A \in \Reals^{n\times m}$, we have $\norm{A} \leq \sqrt{\norm{A}_1 \norm{A}_\infty}$.
\end{proposition}

We can now prove the upper bound:

\begin{lemma}
For all (possibly partial) functions $f$, we have
$\lambda(f)\le\sqrt{\s_0(f)\s_1(f)}$.
\end{lemma}
\begin{proof}
We know that $\lambda(f) = \norm{A_f}$ and $A_f$ is a matrix of the form $\Bigl(\begin{smallmatrix}0& B\\ B^T& 0 \end{smallmatrix}\Bigr)$ if we rearrange the rows and columns so that all $0$-inputs come first and are followed by $1$-inputs, since $A_f$ only connects inputs with different $f$-values. Thus we have
\begin{equation}
    \lambda(f) = \norm{A_f} = \norm{B} \leq \sqrt{\norm{B}_1\norm{B_\infty}} = \sqrt{\s_0(f)\s_1(f)},
\end{equation}
where we used H\"{o}lder's inequality (\Cref{prop:Holder}) and the fact that the maximum row and column sum of $B$ are precisely $\s_0(f)$ and $\s_1(f)$, respectively.
\end{proof}

Our second proof of this claim uses the $\MM_1(f)$ formulation which yields an arguably simpler proof.

\begin{lemma}
For all (possibly partial) functions $f$, we have
$\Adv_1(f)\le\sqrt{\s_0(f)\s_1(f)}$.
\end{lemma}

\begin{proof}
Using the $\MM_1(f)$ version of $\Adv_1(f)$,
set $w(x,i)=\sqrt{\s_0(f)}/\sqrt{\s_1(f)}$ if $f(x)=1$,
and set $w(x,i)=\sqrt{\s_1(f)}/\sqrt{\s_0(f)}$ if $f(x)=0$.
Then if $x$ and $y$ differ in a single bit $i$,
we clearly have $w(x,i)w(y,i)=1$. On the other hand,
$\sum_i w(x,i)\le \s_1(f)\cdot \sqrt{\s_0(f)}/\sqrt{\s_1(f)}
=\sqrt{\s_0(f)\s_1(f)}$
for $1$-inputs $x$, and analogously
$\sum_i w(y,i)\le \sqrt{\s_0(f)\s_1(f)}$
for $0$-inputs $y$.
\end{proof}

Using this better bound on $\lambda(f)$ and Huang's result, we also get that for all total Boolean functions $f$, 
\begin{equation}
    \deg(f) \leq \s_0(f) \s_1(f).
\end{equation}
This result was also recently observed by Laplante, Naserasr, and Sunny~\cite{LNS20}.  Unlike their proof, the following uses Huang's theorem in a completely black-box way.

\begin{proposition}
Assume that $\deg(f) \leq \s(f)^2$ for all total Boolean functions $f$. Then we also have $\deg(f) \leq \s_0(f) \s_1(f)$.
\end{proposition}
\begin{proof}
Let $\s_0(f) = k$ and $\s_1(f) = \ell$. We know that $\deg(f) \leq \max\{k,\ell\}$ by assumption. 
Let $\AND_{k} \circ \OR_{\ell}$ be the AND function on $k$ bits composed with the OR function on $\ell$ bits. Clearly $\s_0(\AND_{k} \circ \OR_{\ell})=\ell$ and $\s_1(\AND_{k} \circ \OR_{\ell}) = k$. 
Furthermore, because the function is monotone, the sensitive bits for a $0$-input are bits set to $0$, and the sensitive bits for a $1$-input are bits set to $1$. 
This means that composing this function with $f$ with yield a function where the one-sided sensitivity will be upper bounded by the product of one-sided sensitivity of the individual functions. Hence for all $b \in \B$, we have
\begin{equation}
    \s_b(\AND_{k} \circ \OR_{\ell} \circ f) \leq \s_b(\AND_{k} \circ \OR_{\ell}) \s_b(f) \leq k\ell.
\end{equation}
Using the assumption on the function $\AND_{k} \circ \OR_{\ell} \circ f$, we get 
\begin{equation}
    \deg(\AND_{k} \circ \OR_{\ell} \circ f) \leq (\s(\AND_{k} \circ \OR_{\ell} \circ f))^2 \leq (k\ell)^2.
\end{equation}
Finally, it is well known that $\deg(f \circ g) = \deg(f)\deg(g)$~(see, e.g., \cite{Tal13}), and hence $\deg(\AND_{k} \circ \OR_{\ell} \circ f) = k\ell \deg(f)$, which implies $\deg(f) \leq k\ell$.
\end{proof}

\subsection{Lower bounds}

Finally, we prove some lower bounds on $\lambda(f)$.

\begin{lemma}\label{lem:lambdalower}
For all (possibly partial) functions $f$, $\s(f)\leq \lambda(f)^2$.
\end{lemma}
\begin{proof}
Consider any input $x$ with sensitivity $\s(f)$. This means $x$ has $\s(f)$ neighbors on the hypercube with different $f$ value. The sensitivity graph restricted to these $\s(f)+1$ inputs is a star graph centered at $x$. The spectral norm of the adjacency matrix of the star graph on $k+1$ vertices is $\sqrt{k}$. Since the spectral norm of $A_f$ is lower bounded by that of a submatrix, we have $\lambda(f) \geq \sqrt{\s(f)}$.
\end{proof}

This relationship is tight for the $\OR_n$ function which has $\s(\OR_n)=n$ and $\lambda(\OR_n)=\sqrt{n}$. Although $\OR_n$ has unbalanced sensitivities, with $\s_0(\OR_n)=n$ and $\s_1(\OR_n)=1$, there are functions $f$ with $\s(f)=\s_0(f)=\s_1(f)=n$ and $\lambda(f)=\sqrt{n}$. One example of such a function is $x_1 \oplus \OR(x_2,\ldots,x_n)$. Another example of such a function with a quadratic gap between $\s(f)$ and $\lambda(f)$ is the function that is $1$ if and only if the input string has Hamming weight $1$. This function has $\s_0(f)=n$ since the all zeros string is fully sensitive and $\s_1(f)=n$ since every Hamming weight $1$ string is also fully sensitive. But we know that this problem can be solved by Grover's algorithm with $O(\sqrt{n})$ queries, and hence $\lambda(f) = O(\Q(f)) = O(\sqrt{n})$.

We can also lower bound $\norm{A_f}$ using the relationship between spectral norm and Frobenius norm. We have for all $N\times N$ matrices $A$ that $\norm{A} \geq \frac{1}{\sqrt{N}} \norm{A}_F$~\cite[Eq. (2.3.7)]{GV13}, where $\norm{A}^2_F = \sum_{i,j}|A_{ij}|^2$. 
For the sensitivity graph of $f$, $\frac{1}{\sqrt{N}} \norm{A_f}_F$ is just the average sensitivity.  
\begin{lemma}
For all (possibly partial) functions $f$, $\lambda(f) \geq \E_{x}[\s_x(f)]$.
\end{lemma}

This can be improved by only taking the expectation on the right over a subset of the inputs of $f$, which then equals another complexity measure originally defined by Khrapchenko~\cite{Khr71}. See \cite{Kou93} for more details. 
\end{document}